\documentclass[preprint]{elsarticle}

\usepackage{array,xspace,multirow,hhline,tikz,colortbl,tabularx,booktabs,fixltx2e,amsmath,amssymb,amsfonts,amsthm}
\usetikzlibrary{arrows}
\usepackage{algorithm}
\usepackage{algorithmic}
\usepackage{eqparbox}
\usepackage{verbatim,ifthen}
\usepackage{enumitem}
\usepackage{pifont}
\usepackage{ifthen}
\usepackage{calrsfs,mathrsfs}
\usepackage{bbding,pifont}
\usepackage{pgflibraryshapes}

	\usepackage{varioref}
		\usepackage{subfigure}

\definecolor{light-gray}{gray}{0.9}

	\newtheorem{lemma}{Lemma}%
	\newtheorem{theorem}{Theorem}%
	\newtheorem{corollary}{Corollary}%
	\newtheorem{example}{Example}
    	
	\newcommand\eat[1]{}

	\usepackage{enumitem}
	\setenumerate[1]{label=\rm(\it{\roman{*}}\rm),ref=({\it\roman{*}}),leftmargin=*}
	\newlength{\wordlength}

	\newcommand{\eqclass}[2][]{\ifthenelse{\equal{#1}{}}{[#2]}{[#2]_{\sim_{#1}}}}

	\newcommand{\squared}[1]{\fbox{#1}}

\usepackage{enumitem}
\setenumerate[1]{label=\rm(\it{\roman{*}}\rm),ref=({\it\roman{*}}),leftmargin=*}

\newcommand{\nbh}[1][]{
	\ifthenelse{\equal{#1}{}}{\nu}{\nu(#1)}
}

\newcommand{\cstr}[1][]{
	\ifthenelse{\equal{#1}{}}{\mathscr S}{\cstr(#1)}
}

\newcommand{\choice}[1][]{
	\ifthenelse{\equal{#1}{}}{\mathit{C}}{\choice(#1)}

		\newcommand{\ml}[1][]{\ensuremath{\ifthenelse{\equal{#1}{}}{\mathit{ML}}{\mathit{ML}(#1)}}\xspace}
		\newcommand{\sml}[1][]{\ensuremath{\ifthenelse{\equal{#1}{}}{\mathit{SML}}{\mathit{SML}(#1)}}\xspace}
		\newcommand{\sd}[1][]{\ensuremath{\ifthenelse{\equal{#1}{}}{\mathit{SD}}{\mathit{SD}(#1)}}\xspace}
		\newcommand{\rsd}[1][]{\ensuremath{\ifthenelse{\equal{#1}{}}{\mathit{RSD}}{\mathit{RSD}(#1)}}\xspace}
		\newcommand{\rd}[1][]{\ensuremath{\ifthenelse{\equal{#1}{}}{\mathit{RD}}{\mathit{RD}(#1)}}\xspace}
		\newcommand{\st}[1][]{\ensuremath{\ifthenelse{\equal{#1}{}}{\mathit{ST}}{\mathit{ST}(#1)}}\xspace}
		\newcommand{\bd}[1][]{\ensuremath{\ifthenelse{\equal{#1}{}}{\mathit{BD}}{\mathit{BD}(#1)}}\xspace}
		\newcommand{\pc}[1][]{\ensuremath{\ifthenelse{\equal{#1}{}}{\mathit{PC}}{\mathit{PC}(#1)}}\xspace}
		\newcommand{\dl}[1][]{\ensuremath{\ifthenelse{\equal{#1}{}}{\mathit{DL}}{\mathit{DL}(#1)}}\xspace}
		\newcommand{\ul}[1][]{\ensuremath{\ifthenelse{\equal{#1}{}}{\mathit{UL}}{\mathit{UL}(#1)}}\xspace}

			\newcommand{\indiff}{\ensuremath{\sim}}}
			
			\usepackage{boxedminipage}
			\usepackage{xspace}

	\newcommand{\vimplies}{\mathrel{\rotatebox{270}{$\implies$}}}	
	
	\newcommand{\vless}{\mathrel{\rotatebox{90}{$<$}}}	
	\newcommand{\vleq}{\mathrel{\rotatebox{90}{$\leq$}}}

	\newcommand{\slantgeq}{\mathrel{\rotatebox{-45}{$\geq$}}}
	
	\newcommand{\slantless}{\mathrel{\rotatebox{-45}{$<$}}}

\begin{document}

		
				
								\title{Achieving Envy-freeness and Equitability with \\Monetary Transfers }
		
		\author{Haris Aziz}\ead{haris.aziz@unsw.edu.au}
		
			\address{UNSW Sydney and Data61 CSIRO, Australia}


	    \begin{abstract}
When allocating indivisible resources or tasks, an envy-free allocation or equitable allocation may not exist. We present a sufficient condition and an algorithm to achieve envy-freeness and equitability when monetary transfers are allowed. The approach works for any agent valuation functions (positive or negative) as long as they satisfy superadditivity. For the case of additive utilities, we present a characterization of allocations that can simultaneously be made equitable and envy-free via payments. 
\end{abstract}

		\maketitle


			\sloppy
		

\section{Introduction}

A fundamental problem that often arises in several settings is that of allocating items, resources, or tasks in a fair manner. 
There are several notions of fairness that have been considered in the literature. Among them, two of the strongest ones are envy-freeess (no agent should envy another agent's outcome) and equitability (every agent should get the same utility). When monetary transfers are not allowed, there may not exist any outcome that is envy-free or equitable. 
This leads to the question: under what conditions fairness can be achieved via monetary transfers?


\paragraph{Contributions}
We consider the situation where we use monetary transfers to achieve both properties simultaneously.
Our results include the following.

We first present a sufficient condition for allocations that can lead to equitability and envy-freeness by monetary transfers. 
In contrast to most of the related results that focus on additive or positive valuations, the statement holds for any superadditive valuations whether they are positive or negative. 
It also leads to a natural and simple algorithm to achieve envy-freeness and equitability. 
The result holds if we replace the payment balance conditions with the condition that agents get subsidies. 

For the domain of additive valuations, we provide a complete characterization of allocations  that can lead to equitability and envy-freeness by monetary transfers. 

We use our insights to design a polynomial-time distributed algorithm that finds an allocation and payment such that the new allocation achieves as much social welfare as a given allocation and the outcome satisfies envy-freeness, equitability, and payment balance. 
Finally, we discuss issues around computation and bounds for minimal payments to achieve fairness. 

%


\section{Related Work}

In the fair division literature (see, e.g. \citep{Aziz20a,BoChLa16,BrTa96a}), envy-freeness~\citep{Fol67} and equitability~\citep{FSVX19a,RoWe97a} are well-known fairness properties. When the items are divisible goods, an equitable and envy-free allocation is guaranteed to exist~\citep{Alon87}. On the other hand, when considering indivisible goods, neither of the two properties are guaranteed to be achievable. 

In this paper, we consider achieving these properties with the help of monetary transfers. 
Fair allocation with money is well-established, especially in the context of room-rent division. A feature of most of the work in the area is that each agent has demand for exactly one item (room)\citep{Ara95,Kli00,Mas87,Su99,Sve83}. 
More general models where envy-freeness is achieved via side-payments have been considered by \citet{HRS02} and \citet{MPR02}. 
\citet{CEM17} consider the distributed allocation of goods and focussed on convergence to envy-free and efficient outcomes via trades among agents.

More recently, there has been focus on computing envy-free allocations when agents have demands for multiple items and monetary transfers are allowed~\citep{HRS02}. In particular, \citet{HS19} popularized the problem of finding allocations for which minimal subsidies will result in envy-freeness.\footnote{The problem of computing minimal subsidies to achieve fairness can be viewed under the framework of ``control of fair division''~\citep{ASW16a} whereby fairness is achieved by minimal modification to the original problem.} 
In followup work, the computational of minimal subsidies has been considered in further depth both from the perspectives of exactly minimal subsidies~\citep{BDN+19a} and approximately minimal subsidies~\citep{CaIo20a}. 
In our model, the valuations can be positive or negative and we additionally target equitability. 
In particular, we use a simple formula for the payment given to each agent. 


\section{Setup}

We consider the setting in which there is a set $N$ of $n$ agents and a set $T$ of $m$ tasks . Each agent $i\in N$ has a valuation function $v_i:2^T \rightarrow \mathbb{R}$. The function $v_i$ specifies a value $v_i(A)$ for a given bundle $A\subseteq T$. The value can be positive or negative. We assume that $v_i(\emptyset)=0$ for all $i\in N$. 
 
The valuation function of an agent $i$ is \emph{supermodular} if for each $i\in N$, and $A,B\subseteq T$,  $v_i(A\cup B)\geq v_i(A)+v_i(B)-v_i(A\cap B)$. 
The valuation function of an agent $i$ is \emph{additive} if for each $i\in N$, and $A,B\subseteq T$ such that $A\cap B=\emptyset$, the following holds: $v_i(A\cup B)= v_i(A)+v_i(B)$.
The valuation function of an agent $i$ is \emph{superadditive} if for each $i\in N$, and $A,B\subseteq T$ such that $A\cap B=\emptyset$, the following holds: $v_i(A\cup B)\geq v_i(A)+v_i(B)$. Note that supermodularity and additivity are stronger conditions than superaddivity.  We assume that valuations satisfy the weaker notion of superaddivity. 


An \emph{allocation} $X=(X_1,\ldots, X_n)$ is a partitioning of the tasks into $n$ bundles where $X_i$ is the bundle allocated to agent $i$.
For an allocation $X$, the social welfare $SW(X)$ is $\sum_{i\in N}v_i(X_i)$.

An \emph{outcome} is a pair consisting of the allocation and the payments made by the agents. Formally, an outcome is a pair $(X,p)$ where $X=(X_1,\ldots X_n)$ is the allocation that specifies bundle $X_i\subseteq T$ for agent $i$ and $p$ specifies the payment $p_i$ made by agent $i$. If $p_i$ is negative, it means agent $i$ gets money. We say that $p$ is \emph{balanced} if $\sum_{i\in N}p_i=0$.

An agent $i$'s \emph{utility} for a bundle-payment pair $(X_j,p_j)$ is $u_i(X_j,p_j)=v_i(X_i)-p_j$. In other words, we assume quasi-linear utilities.
An outcome $(X,p)$  is \emph{envy-free} if for all $i,j\in N$, it holds that $u_i(X_i,p_i)\geq u_i(X_j,p_j)$. An outcome $(X,p)$  is \emph{equitable} if for all $i,j\in N$, $u_i(X_i,p_i)=u_j(X_j,p_j)$.
An allocation $X$ is \emph{envy-freeable} if there exists a payment function $p$ such that $(X,p)$ is envy-free. 
An allocation $X$ is \emph{equitable-convertible} if there exists a payment function $p$ such that $(X,p)$ is equitable. 
An allocation $X$ is \emph{EFEQ-convertible} if there exists a payment function $p$ such that $(X,p)$ is both equitable and envy-free.

For any given allocation $X$, the corresponding envy-graph 
is a complete directed graph with vertex set $N$. For any pair of agents $i,j\in N$ the weight of arc $(i,j)$ 
is the envy agent $i$ has for agent $j$ under the allocation $X$:  $w(i,j) \ =\  v_i(X_{j}) - v_i(X_i)$. For any path or cycle $C$ in the graph, the weight of the $C$ is the sum of weights of arcs along $C$.


\section{Sufficient and necessary conditions to achieve fairness}

We note that every allocation is trivially equitable-convertible: each agent can be given money so that their utility is equal to $\max_{i\in N}v_i(X_i)$. On the other hand, not every allocation is envy-freeable or EFEQ-convertible.

We say that an allocation is \emph{reassignment-stable},
if it maximizes the social welfare across all reassignments of its bundles to
agents. \citet{HS19} assumed positive additive utilities and presented the following elegant characterization of {envy-freeable} allocations. 


\begin{theorem}
	Under positive additive utilities, the following conditions are equivalent for a  given allocation:
	\begin{enumerate}
		\item the allocation is envy-freeable 
		\item the allocation is reassignment-stable
		\item for the allocation, there is no positive weight cycle in the corresponding envy-graph
	\end{enumerate}
	\end{theorem}

The equivalence between the first two conditions has been proved previously (see e.g., \citet{HRS02} and \citet{Mual09a}). Reassignment stability was referred to as local-efficiency by \citet{Mual09a}. 

We explore the conditions under which an allocation is EFEQ-convertible. Firstly, we show that even for positive additive utility, reassignment-stability is not sufficient to simultaneously achieve envy-freeness and equitability via payments.

\begin{example}
Even for positive additive valuations and a given envy-freeable allocation, there may not exist any payments to the agents to achieve both envy-freeness and equitability. Consider an instance with the following additive utilities. We consider an allocation $X$ indicated with the squares in which agent 1 gets $a$ and 2 gets $b$. 

\begin{center}
		\setlength{\tabcolsep}{6pt}
		\begin{tabular}{c|cccccccc}
			& $a$ & $b$\\
			\midrule
			$1$ & \squared{200}&100\\
			$2$ & 2&\squared{1}\\
		\end{tabular}
	\end{center}
	
	The allocation is envy-freeable because it is reassignment-stable. 
We show that there exist no payments to achieve both envy-freeness and equitability simultaneously. Without loss of generality suppose that agents are paid money. The minimum amount needed to obtain equitability is to pay 199 to agent 2. We can maintain equitability by giving equal amounts of money to both the agents. Note however, that the outcome will continue having envy. Agent 1 envies agent 2:
$u_i(X_1,p_1)=200+0 <100+199= u_1(X_2,p_2).$
\qed
	\end{example}  

The example above shows that reassignment-stability is not sufficient to achieve envy-freeness and equitability. Reassignment-stability was the key technique used by \citet{HS19} and \citet{BDN+19a} in their algorithmic results to achieve envy-freeness. 
In our quest to achieve \emph{both} envy-freeness and equitability via monetary transfers, we focus on allocations that are \emph{transfer-stable}. We say that an allocation $X$ is \emph{transfer-stable} if there exist no $i,j\in N$ such that $v_i(X_i\cup X_j)> v_i(X_i)+v_j(X_j)$. We note that under additive valuations, transfer-stability is stronger than the reassignment-stability  property.
   \begin{lemma}\label{lemma-prop}
	 Under additive valuations,  if an allocation is transfer-stable, then it is reassignment-stable.
	  \end{lemma}
	  \begin{proof}
Suppose there exists a reassignment which increases total welfare. This means that the movement of at least one bundle $X_j$ to some agent $i$ increases the social welfare, which implies that the allocation is not transfer-stable. 
\end{proof}

Since transfer-stability is a stronger property than reassignment-stability, a natural question is whether it can be used to achieve stronger fairness guarantees. We answer the question in the affirmative in the following lemma. The lemma applies to the class of superadditive valuations. 

\begin{lemma}\label{lemma:pay}
For a transfer-stable allocation $X$, suppose each agent $i$ makes a payment equal to $p_i=v_i(X_i)-SW(X)/n$. Then if agent valuations are superadditive, the outcome $(X,p)$  is envy-free and equitable.  
	\end{lemma}
\begin{proof}
	We first want to prove envy-freeness:  for all $i,j\in N$, it holds that $u_i(X_i,p_i)\geq u_i(X_j,p_j)$.
By transfer-stability of allocation $X$, 
\begin{align*}
v_i(X_i)+v_j(X_j)\geq v_i(X_i\cup X_j).
\end{align*}

Since $v_i$ is superadditive, it follows that 
\begin{align*}
v_i(X_i\cup X_j)\geq v_i(X_i) + v_i(X_j).
\end{align*}

By combining the two inequalities above, we get
\begin{align*}
&v_i(X_i)+v_j(X_j)\geq v_i(X_i) + v_i(X_j)\\
\iff&v_j(X_j)\geq v_i(X_j)\\
\iff&0\geq v_i(X_j)-v_j(X_j)\\
\iff&v_i(X_i)-v_i(X_i)+SW(X)/n\\
&\geq v_i(X_j)-v_j(X_j)+SW(X)/n\\
\iff&v_i(X_i)-(v_i(X_i)-SW(X)/n)\\
&\geq v_i(X_j)-(v_j(X_j)-SW(X)/n)\\
\iff&v_i(X_i)-p_i\geq v_i(X_j)-p_j\\
\iff&u_i(X_i,p_i)\geq u_i(X_j,p_j).
\end{align*}
The last inequality indicates that agent $i$ is not envious of $j$ and hence $(X,p)$ satisfies envy-freeness. 

Next, we argue that the outcome $(X,p)$ satisfies equitability. Each agent $i\in N$ gets utility $v_i(X_i)-p_i=v_i(X_i)- (v_i(X_i)-SW(X)/n)=SW(X)/n$. Since each agent has the same utility $SW(X)/n$, the outcome satisfies equitability. 
\end{proof}

The payment function $p_i=v_i(X_i)-SW(X)/n$ used in the lemma is not new. It is referred to as the \emph{Knaster} payments~\citep{Knas46a} and is inspired by the idea that each agent should get utility that is at least the proportionality guarantee $v_i(T)/n$ that was popularized by \citet{steinhaus1948problem}.
In the literature on fair allocation with money, Knaster payments have typically been applied on welfare maximizing allocations. \citet{Rait00a} discusses them prominently in the context of 2 agents and additive valuations. We show that it is sufficient to consider superadditive valuations and transfer-stable allocations for Knaster payments to achieve both equitability and envy-freeness.

Our insights also show that any social welfare maximizing allocation is EFEQ-convertible. 
\begin{corollary}
	For superadditive utilities, a social welfare maximizing allocation is EFEQ-convertible. 
	\end{corollary}
	\begin{proof}
		A social welfare maximizing allocation is transfer-stable. By Lemma~\ref{lemma:pay}, it is EFEQ-convertible. 
		\end{proof}

In Lemma~\ref{lemma:pay}, we have shown that for (super)additive valuations, transfer-stability is a sufficient condition to simultaneously achieve equitability and envy-freeness via payments. Next, we show that transfer-stability is also a necessary condition. 

\begin{lemma}\label{lemma:eqeq}
Under additive utilities, if an allocation is EFEQ-convertible, then it is transfer-stable. 
\end{lemma}

\begin{proof}
Suppose an allocation $X$ is not transfer-stable.
Then there exist agents $i,j\in N$ such that
\begin{align*}
	v_i(X_j)> v_j(X_j)
	\end{align*}
	
	The inequality is depicted in Figure~\ref{fig1}.
	
	   							\begin{figure}[h!]
	   								\begin{center}
	  \scalebox{0.9}{ 						\begin{tikzpicture}
	   							\tikzstyle{pfeil}=[->,>=angle 60, shorten >=1pt,draw]
	   							\tikzstyle{onlytext}=[]

	\node        (i) at (0,0) {$i$};
	\node        (xi) at (1,0.5) {$X_i$};
	\node        (xj) at (3,0.5) {$X_j$};
	\node        (ixi) at (1,0) {$v_i(X_i)$};
	\node        (ixj) at (3,0) {$v_i(X_j)$};

	\node        (i) at (0,-2) {$j$};
	\node        (jxi) at (1,-2) {$v_j(X_i)$};
	\node        (jxj) at (3,-2) {$v_j(X_j)$};

	\node        (ixj) at (3,-1) {$\color{red}{\vless}$};

	   						\end{tikzpicture}
							}
	   						\end{center}
	   						\caption{A case in the proof of Lemma~\ref{lemma:eqeq}}
							\label{fig1}
	   						\end{figure}
	
	If $X$ is not envy-freeable, we are done so we assume that $X$ is envy-freeable. 
	Then it must be that 
	\begin{align*}
		v_j(X_i)\leq v_i(X_i)
		\end{align*}
	or we can swap the allocations of $i$ and $j$ to get a welfare improvement which means that is not envy-freeable which implies that it is not EFEQ-convertible.
	The case is depicted in Figure~\ref{fig2}.
	
	   							\begin{figure}[h!]
	   								\begin{center}
	  \scalebox{0.9}{ 						\begin{tikzpicture}
	   							\tikzstyle{pfeil}=[->,>=angle 60, shorten >=1pt,draw]
	   							\tikzstyle{onlytext}=[]

	\node        (i) at (0,0) {$i$};
	\node        (xi) at (1,0.5) {$X_i$};
	\node        (xj) at (3,0.5) {$X_j$};
	\node        (ixi) at (1,0) {$v_i(X_i)$};
	\node        (ixj) at (3,0) {$v_i(X_j)$};

	\node        (i) at (0,-2) {$j$};
	\node        (jxi) at (1,-2) {$v_j(X_i)$};
	\node        (jxj) at (3,-2) {$v_j(X_j)$};

	\node        () at (3,-1) {$\vless$};
	
		\node        () at (1,-1) {$\color{red}{\vleq}$};

			
	   						\end{tikzpicture}
							}
	   						\end{center}
	   						\caption{A case in the proof of Lemma~\ref{lemma:eqeq}}
							\label{fig2}
	   						\end{figure}

	\medskip
By the characterization result of \citet{HS19}, we know that $X$ does not admit an envy-cycle. Therefore, either 

\begin{align*}
v_i(X_i)\geq v_i(X_j)
	\end{align*}
or 
\begin{align*}
v_j(X_j)\geq v_j(X_i)
	\end{align*}

We first consider the case 	
$v_i(X_i)\geq v_i(X_j)$ which is depicted in Figure~\ref{fig3}.

	   							\begin{figure}[h!]
	   								\begin{center}
	  \scalebox{0.9}{ 						\begin{tikzpicture}
	   							\tikzstyle{pfeil}=[->,>=angle 60, shorten >=1pt,draw]
	   							\tikzstyle{onlytext}=[]

	\node        (i) at (0,0) {$i$};
	\node        (xi) at (1,0.5) {$X_i$};
	\node        (xj) at (3,0.5) {$X_j$};
	\node        (ixi) at (1,0) {$v_i(X_i)$};
	\node        (ixj) at (3,0) {$v_i(X_j)$};

	\node        (i) at (0,-2) {$j$};
	\node        (jxi) at (1,-2) {$v_j(X_i)$};
	\node        (jxj) at (3,-2) {$v_j(X_j)$};

	\node        () at (3,-1) {$\vless$};
	
		\node        () at (1,-1) {$\vleq$};
		
		\node        () at (2,0) {$\color{red}{\geq}$};

	   						\end{tikzpicture}
							}
	   						\end{center}
	   						\caption{A case in the proof of Lemma~\ref{lemma:eqeq}}
							\label{fig3}
	   						\end{figure}

Since $v_i(X_j)> v_j(X_j)$, it follows that 
\begin{align*}
v_i(X_i)\geq v_i(X_j) >v_j(X_j).
	\end{align*}
	
	Since $i$ gets a strictly higher value than $j$ from her allocation, we need to pay money to agent $j$ to ensure equitability. In particular, agent $j$ is paid amount $v_i(X_i)-v_j(X_j)$. In that case agent $i$'s estimation of agent $j$'s outcome is $v_i(X_j)+(v_i(X_i)-v_j(X_j))$ where we know that $v_i(X_i)-v_j(X_j)>0$. Therefore agent $i$ is envious of agent $j$. Hence $X$ is not EFEQ-convertible.
	
	   							\begin{figure}[h!]
	   								\begin{center}
	  \scalebox{0.9}{ 						\begin{tikzpicture}
	   							\tikzstyle{pfeil}=[->,>=angle 60, shorten >=1pt,draw]
	   							\tikzstyle{onlytext}=[]

	\node        (i) at (0,0) {$i$};
	\node        (xi) at (1,0.5) {$X_i$};
	\node        (xj) at (3,0.5) {$X_j$};
	\node        (ixi) at (1,0) {$v_i(X_i)$};
	\node        (ixj) at (3,0) {$v_i(X_j)$};

	\node        (i) at (0,-2) {$j$};
	\node        (jxi) at (1,-2) {$v_j(X_i)$};
	\node        (jxj) at (3,-2) {$v_j(X_j)$};

	\node        () at (3,-1) {$\vless$};
	
		\node        () at (1,-1) {$\vleq$};

 		\node        () at (2,-2) {$\color{red}{\leq}$};
			
	   						\end{tikzpicture}
							}
	   						\end{center}
	   						\caption{A case in the proof of Lemma~\ref{lemma:eqeq}}
							\label{fig4}
	   						\end{figure}
	
	In Figure~\ref{fig3}, we assumed that $v_j(X_j)< v_j(X_i)$. Next we consider the other case $v_j(X_j)\geq v_j(X_i)$ which is depicted in Figure~\ref{fig4}. 

	\medskip
We distinguish between two cases (a) $v_i(X_i)\geq v_i(X_j)$ and (b) $v_i(X_i)< v_i(X_j)$.
	
	We already considered case (a) $v_i(X_i)\geq v_i(X_j)$ in the previous analysis (Figure~\ref{fig3}). Therefore, we now consider case (b) and assume that $v_i(X_i)< v_i(X_j)$ which is depicted in Figure~\ref{fig5}.

	   							\begin{figure}[h!]
	   								\begin{center}
	  \scalebox{0.9}{ 						\begin{tikzpicture}
	   							\tikzstyle{pfeil}=[->,>=angle 60, shorten >=1pt,draw]
	   							\tikzstyle{onlytext}=[]

	\node        (i) at (0,0) {$i$};
	\node        (xi) at (1,0.5) {$X_i$};
	\node        (xj) at (3,0.5) {$X_j$};
	\node        (ixi) at (1,0) {$v_i(X_i)$};
	\node        (ixj) at (3,0) {$v_i(X_j)$};

	\node        (i) at (0,-2) {$j$};
	\node        (jxi) at (1,-2) {$v_j(X_i)$};
	\node        (jxj) at (3,-2) {$v_j(X_j)$};

	\node        () at (3,-1) {$\vless$};
	
		\node        () at (1,-1) {$\vleq$};

 		\node        () at (2,-2) {$\color{black}{\leq}$};
		
		 \node        () at (2,0) {$\color{red}{<}$};
			
	   						\end{tikzpicture}
							}
	   						\end{center}
	   						\caption{A case in the proof of Lemma~\ref{lemma:eqeq}}
							\label{fig5}
	   						\end{figure}

We distinguish between two further final cases: case $v_i(X_i)\geq v_j(X_j)$ and the case $v_i(X_i)< v_j(X_j)$.
\begin{enumerate}
	\item $v_i(X_i)\geq v_j(X_j)$ which is depicted in Figure~\ref{fig6}.

	   							\begin{figure}[h!]
	   								\begin{center}
	  \scalebox{0.9}{ 						\begin{tikzpicture}
	   							\tikzstyle{pfeil}=[->,>=angle 60, shorten >=1pt,draw]
	   							\tikzstyle{onlytext}=[]

	\node        (i) at (0,0) {$i$};
	\node        (xi) at (1,0.5) {$X_i$};
	\node        (xj) at (3,0.5) {$X_j$};
	\node        (ixi) at (1,0) {$v_i(X_i)$};
	\node        (ixj) at (3,0) {$v_i(X_j)$};

	\node        (i) at (0,-2) {$j$};
	\node        (jxi) at (1,-2) {$v_j(X_i)$};
	\node        (jxj) at (3,-2) {$v_j(X_j)$};

	\node        () at (3,-1) {$\vless$};
	
		\node        () at (1,-1) {$\vleq$};

 		\node        () at (2,-2) {$\color{black}{\leq}$};
		
		 \node        () at (2,0) {$\color{black}{<}$};
		 
		 \node        () at (2,-1) {$\color{red}{\slantgeq}$};
			
	   						\end{tikzpicture}
							}
	   						\end{center}
	   						\caption{A case in the proof of Lemma~\ref{lemma:eqeq}}
							\label{fig6}
	   						\end{figure}

	Since $v_i(X_i)< v_i(X_j)$, agent $1$ is envious of agent $2$ and needs money to remove the envy. On other hand, we know that $v_i(X_j)> v_j(X_j)$ so agent $2$ needs more money to achieve equitability. Both the properties cannot be met. 
	\item $v_i(X_i)< v_j(X_j)$ which is depicted in Figure~\ref{fig7}.
	
	   							\begin{figure}[h!]
	   								\begin{center}
	  \scalebox{0.9}{ 						\begin{tikzpicture}
	   							\tikzstyle{pfeil}=[->,>=angle 60, shorten >=1pt,draw]
	   							\tikzstyle{onlytext}=[]

	\node        (i) at (0,0) {$i$};
	\node        (xi) at (1,0.5) {$X_i$};
	\node        (xj) at (3,0.5) {$X_j$};
	\node        (ixi) at (1,0) {$v_i(X_i)$};
	\node        (ixj) at (3,0) {$v_i(X_j)$};

	\node        (i) at (0,-2) {$j$};
	\node        (jxi) at (1,-2) {$v_j(X_i)$};
	\node        (jxj) at (3,-2) {$v_j(X_j)$};

	\node        () at (3,-1) {$\vless$};
	
		\node        () at (1,-1) {$\vleq$};

 		\node        () at (2,-2) {$\color{black}{\leq}$};
		
		 \node        () at (2,0) {$\color{black}{<}$};
		 
		 \node        () at (2,-1) {$\color{red}{\slantless}$};
			
	   						\end{tikzpicture}
							}
	   						\end{center}
	   						\caption{A case in the proof of Lemma~\ref{lemma:eqeq}}
							\label{fig7}
	   						\end{figure}
	
	Since $v_i(X_i)< v_i(X_j)$, agent $1$ is envious of agent $2$ and needs money to remove the envy. The exact amount needed to remove envy is $v_i(X_j)-v_i(X_i)$. But then the new utility of agent $i$ is $v_i(X_j)$ which we know (see Figure~\ref{fig7}) is more than $v_j(X_j)$ so equitability is violated. 
\end{enumerate}

We have proved that in all the cases, if an allocation is not transfer-stable, then it is not EFEQ-convertible.
\end{proof}



We obtain the following result: transfer-stability characterizes EFEQ-convertible allocations. 

\begin{theorem}\label{th:charac}
Under additive utilities, an allocation is EFEQ-convertible if and only if it is transfer-stable. 
\end{theorem}
\begin{proof}
	The statement follows from Lemma~\ref{lemma:pay} and Lemma~\ref{lemma:eqeq}.
	\end{proof}
	
	\begin{corollary}
		Assuming that we have access to an oracle that gives the utility of an agent for a bundle in constant time, then for additive valuations, there exists a $O(n^2)$ algorithm to check whether a given allocation is EFEQ-convertible.
		\end{corollary}
		\begin{proof}
			By Theorem~\ref{th:charac}, we need to check whether the allocation is transfer-stable or not.  \end{proof}
	
	\section{An algorithm to achieve fairness with payments}

	The following lemma shows that a greedy distributed approach can achieve a transfer-stable allocation. 

	\begin{lemma}\label{lemma:local}
	Suppose there exists an oracle that computes the value of an agent for a bundle of tasks in time $f(I)$. Then, for any given allocation $X$, a transfer-stable allocation can be computed in $O(n^4f(I))$ such that $SW(Y)\geq SW(X)$.
		\end{lemma}
	\begin{proof}
		We take any pair of agents $i,j\in N$ and check if $v_i(X_i\cup X_j)> v_i(X_i)+v_j(X_j)$. This can be checked in time $O(f(I)m)$ for a pair of agents and in 
$O(f(I)n^2)$ for all pairs of agents. 
	If $v_i(X_i\cup X_j)> v_i(X_i)+v_j(X_j)$, we give the allocation of $j$ to agent $i$ which results in agent $j$ getting an empty bundle. With each such operation the total social welfare increases. Hence, we the process terminates. Next, we prove that the process terminates in a polynomial number of steps. 
	
	With each operation, one of the two cases occurs.
The first case is that an additional agent $j$ completely loses her bundle. 
When a bundle going to another agent who has a non-empty bundle, then the number of agents who have an empty bundle increases. Such operations can happen at most $n-1$ times.  Now suppose that the the number of agents who have an empty bundle does not increase. This is only possible in the case that that agent $i$ had an empty allocation who gets the bundle $X_j$. Since each transfer of a bundle is welfare improving, it cannot happen that a bundle is returned to an agent $i$. Therefore such operations can happen at most $n-1$ times until the bundle will not move to any agent with an empty allocation.

	Hence these operations can happen at most $n^2$ times until no more transfers are possible. 
		\end{proof}

	Lemma~\ref{lemma:pay} and Lemma~\ref{lemma:local} give us an easy constructive method to achieve envy-freeness and equitability. 
	The method is presented as Algorithm~\ref{algo:EF}.
	We use the algorithm in the proof of Lemma~\ref{lemma:local} to obtain a transfer-stable allocation.\footnote{In practice in several domains such as routing-based task allocation, one would expect most reasonable and balanced allocations to be transfer-stable.}
	After that we use the payment function specified in Lemma~\ref{lemma:pay} to achieve envy-freeness and equitability.  Algorithm~\ref{algo:EF} leads to the following theorem.

										\begin{algorithm}[h!]
											  \caption{Envy-freeness and equitability with payments}
											  \label{algo:EF}
			\normalsize
											\begin{algorithmic}
												\REQUIRE  Allocation $Y$ and valuations functions $v_i$ for each agent $i\in N$. 							\ENSURE Allocation $X$ and payment function $p$
											\end{algorithmic}
											\begin{algorithmic}[1]
												\normalsize
			 \STATE Allocation $X\longleftarrow Y$
			 \WHILE{there exists some $i,j\in N$ s.t. $v_i(X_i\cup X_j)> v_i(X_i)+v_j(X_j)$}
			 \STATE $X_i\longleftarrow X_i\cup X_j$
			 \STATE $X_j\longleftarrow \emptyset$
			 \ENDWHILE
			 \STATE For each agent $i$, $p_i\longleftarrow v_i(X_i)-SW(X)/n$
			 \STATE $p_i'\longleftarrow p_i-\max \{p_j: j\in N, p_j>0\}$
				
												\RETURN $(X,p)$ where $p$ is balanced or $(X,p')$ where $p'$ is negative.
											\end{algorithmic}
										\end{algorithm}

	\begin{theorem}
	Suppose agents have superadditive valuations. Then for a given allocation $Y$, an allocation $X$ and payment function $p$ can be computed in polynomial time such that 
	\begin{enumerate}
		\item the outcome $(X,p)$ is equitable and envy-free, 
		\item $SW(X)\geq SW(Y)$, and
		\item $p$ is balanced.
	\end{enumerate}

		\end{theorem}

	Note that our result allows for some payments to be positive, i.e., some agents need to pay money. If we insist on simply using subsidies from a third party to achieve envy-freeness, then we can find the largest payment $p_i'$ made by an agent $i$ and give each agent an additional amount of $p_i'$ so that agents only get money and do not need to give money. To be precise, if the balanced payment is $p$, we can get negative or zero payments $p'$ as follows: $p_i'$ is set to $p_i-\max \{p_j: j\in N, p_j>0\}$. By doing this, we obtain the next theorem. 
	
		\begin{theorem}
		Suppose agents have superadditive valuations. Then for a given allocation $Y$, an allocation $X$ and payment function $p'$ can be computed in polynomial time  such that 
		\begin{enumerate}
			\item the outcome $(X,p')$ is equitable and envy-free, 
			\item $SW(X)\geq SW(Y)$, and
			\item for each $i\in N$, $p_i'\leq 0$.
		\end{enumerate}

			\end{theorem}

Note that if we are not given an initial allocation, then we can achieve an outcome satisfying envy-freeness and equitability in an even simpler way. We bundle all the tasks together and then give the bundle $T$ to an agent $i$ for which value $v_i(T)$ is the highest. By construction, the allocation is transfer-stable. We then implement the payment function as specified in Algorithm~\ref{algo:EF}.

\section{Minimal payments to achieve fairness}

When using payments to achieve fairness, one may want to use the \emph{minimal} exchange of money or subsidy to achieve fairness. The problem has been explored by \citet{HS19} and \citet{BDN+19a} when the goal is envy-freeness.

Suppose a given allocation $X$ is EFEQ-convertible. Then there is a linear-time algorithm to compute the minimal payments to achieve both envy-freeness and equitability. The key insight is that for 
EFEQ-convertible allocations, it is sufficient to simpy focus on achieving equitability. Any additional and uniform payment for all agents does not affect envy-freeness. Therefore, we can give agents sufficient money to ensure that each agent has utility equal to $\max_{i\in N}v_i(X_i)$. Next, we consider the problem in which we can choose a suitable allocation so as to require minimal payments to acheive fairness.

\begin{theorem}
	Computing the minimum payments to simultaneously achieve envy-freeness and equitability is strongly NP-hard. Unless $\text{P}=\text{NP}$, there exists no deterministic polynomial-time algorithm that approximates within any given positive factor the minimum payments to simultaneosly achieve envy-freeness and equitability.
	\end{theorem}
The proof follows from the fact that checking whether there exists an envy-free allocation (that requires zero payments to achieve envy-freeness) is NP-complete if the agents have  identical valuations. Under identical valuations, envy-freenss also implies equitability. The inapproximability result follows from the fact than even checking whether zero payment is required is NP-hard.

	   							\begin{figure}[t!]
	   								\begin{center}
	  \scalebox{0.9}{ 						\begin{tikzpicture}
	   							\tikzstyle{pfeil}=[->,>=angle 60, shorten >=1pt,draw]
	   							\tikzstyle{onlytext}=[]

	\node        (EQEQ) at (0,2) {EFEQ-convertible};
	
		\node        () at (3,2) {$\iff$};
		
			\node        () at (0,1) {$\vimplies$};

				\node        () at (6,1) {$\vimplies$};

	\node        (EQ) at (0,0) {envy-freeable};
	
		\node        () at (3,0) {$\iff$};
	
	\node        (EQEQ) at (6,2) {transfer-stable};

	\node        (EQ) at (6,0) {reassignment-stable};
	
			\node        () at (0,-1) {$\vimplies$};
	
	\node        (EQ) at (0,-2) {equitable-convertible};
	\node        () at (3,-2) {$\iff$};
	
		\node        () at (6,-1) {$\vimplies$};
	
		\node        (EQ) at (6,-2) {no restriction};

	   						\end{tikzpicture}
							}
	   						\end{center}
	   						\caption{Properties of allocations under additive valuations}
							\label{fig:relations}
	   						\end{figure}
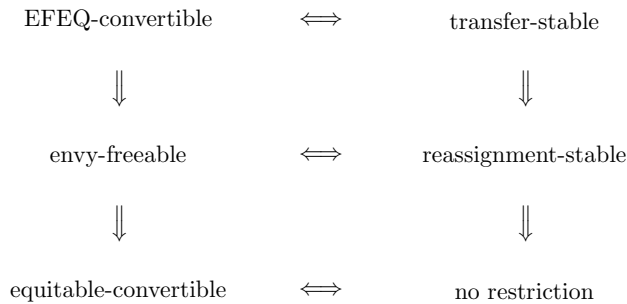

\medskip

In order to achieve reasonable bounds on the maximum subsidy to achieve envy-freeness, \citet{HS19} and \citet{BDN+19a} assume valuations in which $v_i(t)\leq 1$ for all $i\in N$ and $t\in T$. Our goal is to achieve envy-freeness and equitability simultaneously. 
We first note that we inherit any lower bounds on subsidies required to get envy-freeness. Therefore, when an EFEQ-convertible allocation is given, the minimum subsidy required is $(n-1)m$ in the worst case. We also observe that at least $(n-1)m$ payment may be required even when an EFEQ-convertible allocation is \emph{not} given and we can choose an EFEQ-convertible allocation intelligently. The reason is that in order to achieve transfer-stability, it may be the case that all the items need to be given to the same agent.

\section{Conclusions}

Achieving fairness via payments is an interesting reseach direction. In this paper, we focussed on envy-freeness and equitability and presented a characterization of allocations that are EFEQ-convertible allocation.

Figure~\ref{fig:relations} highlights some of the insights from this paper and the paper of \citet{HS19}. 
It will be interesting to explore other desirable fairness properties.

		\section*{Acknowledgements}
		Aziz is supported by a UNSW Scientia Fellowship, and Defence Science and Technology (DST) under the project ``Auctioning for distributed multi vehicle planning'' (DST 9190). He thanks Ioannis Caragiannis, Alex Lam  and Bo Li for helpful comments and pointers.
		%

		\end{document}